\documentclass{amsart}
\usepackage{amsmath}
\usepackage{amsthm}
\usepackage{amsfonts}
\usepackage{amssymb}
\usepackage{amsbsy}
\usepackage[dvips]{graphicx}
\usepackage{amsaddr}

\newtheorem{theorem}{Theorem}
\newtheorem{proposition}[theorem]{Proposition}

\theoremstyle{definition}
\newtheorem{problem}{Problem}

\newtheorem{remark}{Remark}

\newcommand{\beq}{\begin{equation}}
\newcommand{\eeq}{\end{equation}}

\newcommand{\HH}{\mathbb{H}}

\newcommand{\cZ}{\mathcal{Z}}

\newcommand{\rmd}{\mathrm {d}}
\newcommand{\rmi}{\mathrm {i}}

\newcommand{\sds}{\strut\displaystyle}
\def\staccrel#1#2{\mathrel{\mathop{#1}\limits_{#2}}}

\newcommand{\rmE}{\mathrm{E}\,}
\newcommand{\rmR}{\mathrm{R}\,}

\newcommand{\N}{\mathbb{N}}
\newcommand{\I}{\mathbb{I}}

\newcommand{\cE}{\mathcal{E}}
\newcommand{\cI}{\mathcal{I}}
\newcommand{\cN}{\mathcal{N}}

\newcommand{\txi}{\widetilde{\xi}}
\newcommand{\teta}{\widetilde{\eta}}
\newcommand{\tpsi}{\widetilde{\psi}}

\newcommand{\og}{\mathfrak{g}}
\newcommand{\of}{\mathfrak{f}}
\newcommand{\oJ}{\mathfrak{J}}

\begin{document}

\title[Inverse optical imaging]
{Inverse optical imaging viewed as a backward channel communication problem}

\author[E. De Micheli]{Enrico De Micheli}
\address{\sl IBF -- Consiglio Nazionale delle Ricerche \\ Via De Marini, 6 - 16149 Genova, Italy}
\email{enrico.demicheli@cnr.it}

\author[G. A. Viano]{Giovanni Alberto Viano}
\address{\sl Dipartimento di Fisica -- Universit\`a di Genova,\\
Istituto Nazionale di Fisica Nucleare -- Sezione di Genova, \\
Via Dodecaneso, 33 - 16146 Genova, Italy}
\email{viano@ge.infn.it}

\begin{abstract}
The inverse problem in optics, which is closely related to the classical
question of the \emph{resolving power}, is reconsidered as a communication channel
problem. The main result is the evaluation of the maximum number $M_\varepsilon$
of $\varepsilon$--distinguishable messages ($\varepsilon$ being a bound
on the noise of the image) which can be conveyed back from the image to reconstruct the object.
We study the case of coherent illumination. By using the concept of
Kolmogorov's $\varepsilon$--capacity, we obtain:
$M_\varepsilon \sim 2^{\,S\,\log(1/\varepsilon)}\xrightarrow[\varepsilon\to 0] {} \infty$,
where $S$ is the \emph{Shannon number}. Moreover, we show that the
$\varepsilon$--capacity in inverse optical imaging is nearly equal to the amount
of information on the object which is contained in the image. We thus compare
the results obtained through the classical information theory, which is
based on the probability theory, with those derived from
a form of \emph{topological information theory}, based on
Kolmogorov's $\varepsilon$--entropy and $\varepsilon$--capacity, which are
concepts related to the evaluation of the \emph{massiveness}
of compact sets.
\end{abstract}


\maketitle


\section{Introduction}
\label{se:introduction}
The definition of the \emph{resolving power} of an optical system is a classical problem of optics with a very
long history, which goes back to Lord Rayleigh. It is precisely his criterion for resolution
which is a milestone in this theory. As is well--known, however, this criterion remains
somehow empirical, and it is sometimes considered a \emph{quite arbitrary choice}.

According to geometrical optics, the image of a point source provided by an optical
instrument is a perfectly sharp point. However, because of diffraction effects,
the image of a point is not a point but a small light patch, called the diffraction pattern.
Optical instruments, whose diffraction effects are important, are called diffraction--limited
imaging systems, and hereafter we shall refer to only this type of optical systems.

We assume, for simplicity, that the scalar theory of light can be used. In this theory
monochromatic light is represented by a scalar function, which is usually written as a complex--valued
function, called the complex amplitude, whose modulus and phase are respectively the amplitude and the
phase of the light disturbance. In the case of spatially coherent illumination
(for short, coherent illumination) the relative phase of two object points
is constant in time, i.e. even if the two phases can vary randomly in time, they vary in
an identical fashion.

We consider systems producing real (non virtual) images. We also assume that the system is
isoplanatic, i.e., space--invariant. In practice, optical imaging systems are seldom isoplanatic
over the whole object field, but it is also possible to divide the object field into regions
within which the system is approximately space--invariant. Finally, we assume that the
magnification factor of the optical system has been reduced to one by a suitable re--scaling of the
space variables of the image plane.

Diffraction--limited imaging systems are usually treated by Fourier methods, and the
corresponding theory is called Fourier optics. Assume that $f(x)$ denote the complex
amplitude distribution of a coherently illuminated object; for reasons of simplicity
but without loss of generality, we limit ourselves to consider unidimensional objects.
The Fourier transform of $f(x)$,
\beq
F(\omega) = \frac{1}{\sqrt{2\pi}}\int_{-\infty}^{+\infty} f(x)\,e^{-\rmi\omega x}\,\rmd x,
\label{1.1}
\eeq
is an entire function in the complex $\omega$--plane since $f(x)$ is
space--limited. Then one could argue (as observed by several authors
\cite{Wolter,Toraldo}) that, even though the knowledge of the function $F(\omega)$ is limited to
the finite interval $|\omega|\leqslant\Omega$ since the pupil stops all the waves with $\omega$
larger than a positive constant $\Omega$, nevertheless, in view of the uniqueness of the analytic continuation,
one could determine uniquely $F(\omega)$ everywhere. Hence, the object could be reconstructed
in all its details, and there should be no loss of information in passing through the optical system: in principle,
analytic continuation in the frequency domain will allow for restoration
of unlimited details \cite{Wolter}. But the uniqueness of the analytic continuation
does not imply its stability, namely, a continuous dependence of the solution on the data.
The ill--posedness \cite{Hadamard} of the analytic continuation, and more generally of the
inverse problem, was then recognized \cite{Viano1}, and the theory of the regularization of
the ill--posed problems in the sense of Hadamard was extensively applied to Fourier optics \cite{Bertero1}.

\begin{figure}[tb] 
\begin{center}
\leavevmode
\includegraphics[scale=0.4]{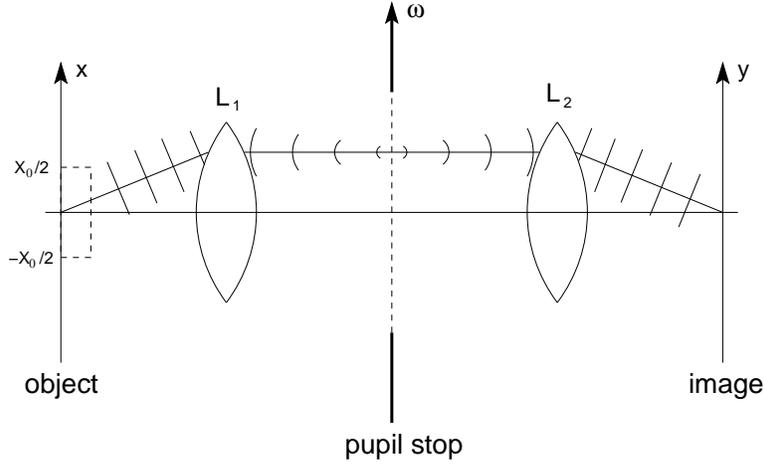}
\caption{\label{fig:1} Schematic of coherent light image formation in a one--dimensional diffraction--limited optical
system (see also \cite{Toraldo}).} 
\end{center}
\end{figure}

The mathematical inverse problem in optics, within the scheme outlined above,
can be formulated as follows. Consider an unidimensional object and refer to the conventional
optical system depicted in Fig. \ref{fig:1}. A plane object, illuminated with coherent light,
gives rise to a complex amplitude distribution $f(x)$ at the front focal plane of the lens $L_1$ (see Fig. \ref{fig:1}).
A real image is formed at the rear focal plane of the lens $L_2$. Lenses $L_1$ and $L_2$ have a common
focus at the stop plane or pupil plane. Now, let us return to the Fourier transform $F(\omega)$;
if both lenses $L_1$ and $L_2$ are assumed to fulfill the sine condition \cite[p. 166]{Born},
then $\omega$ is proportional to the vertical coordinate on the pupil plane (see Fig. \ref{fig:1}).
As we have already remarked, the pupil--stop blocks all of the contributions that have $|\omega|$ larger than the positive
constant $\Omega$. As a consequence, on the image plane we will not recover exactly $f(x)$, but its
band--limited version
\beq
g(y) = \frac{1}{\sqrt{2\pi}}\int_{-\Omega}^{\Omega}F(\omega)\,e^{\rmi\omega y}\,\rmd\omega.
\label{2.1}
\eeq
Inserting the expression of $F(\omega)$ given by \eqref{1.1} into \eqref{2.1},
and assuming that the object distribution $f(x)$ vanishes outside the interval
$-X_0/2 \leqslant x \leqslant X_0/2$, we have
\beq
g(y) = \frac{1}{2\pi}\int_{-\Omega}^\Omega e^{\rmi\omega y}\,\rmd\omega
\int_{-X_0/2}^{X_0/2} f(x)\,e^{-\rmi\omega x}\,\rmd x =
\int_{-X_0/2}^{X_0/2}\frac{\sin[\Omega(x-y)]}{\pi(x-y)}\,f(x)\,\rmd x.
\label{3.1}
\eeq
The image function $g(y)$ is an entire band--limited function, and the sampling theorem guarantees
that it can be reconstructed, without loss of information, when its values are known at a set of sampling
points, chosen in arithmetic progression with difference $\pi/\Omega$;
notice that in optics, in the unidimensional situation, the Rayleigh distance $R$
(also called the \emph{Nyquist distance}) is $R=(\pi/\Omega)=\mathrm{resolution ~distance}$.
In particular, the image $g(y)$ can be reconstructed in the interval $(-X_0/2,X_0/2)$ from the knowledge of the function
on a set of $S$ points, where $S \doteq X_0/(\frac{\pi}{\Omega})=\Omega X_0/\pi$ is the \emph{Shannon number}
of the image \cite{Toraldo,Stern}. It is in this connection that several authors (notably, Toraldo di Francia \cite{Toraldo})
argued that an image can be completely determined by $S$ (complex) numbers, which are called the \emph{image degrees of freedom}.

Equation \eqref{3.1} can be re--written in operator form as follows
\beq
(Af)(y) = \int_{-X_0/2}^{X_0/2} \frac{\sin[\Omega(x-y)]}{\pi(x-y)}\,f(x)\,\rmd x = g(y)
\qquad \left(-\frac{X_0}{2}\leqslant y \leqslant \frac{X_0}{2}\right).
\label{4.1}
\eeq
Then, the problem of object restoration is equivalent to solving the Fredholm
integral equation of the first kind $Af=g$, where $A$ is a self--adjoint, non--negative and compact operator,
$g$ represents the data (the image), and $f$ is the unknown (the object distribution).
As we remarked above, this problem is ill--posed: the solution to Eq. \eqref{4.1}, even if it
is unique, does not depend continuously on the data. Small perturbations of the data, due to the noise,
produce wide oscillations in the solution, the problem needs regularization.

Summarily, we may distinguish two different approaches to regularization:
\begin{itemize}
\item[(a)] Methods that require well--defined \emph{a priori} global bounds on the solution,
and work in definite functional spaces. These methods could be called \emph{deterministic},
understanding this terminology in a wide sense \cite{Tikhonov,Groetsch}.
\item[(b)] Methods that make use of techniques taken from the theory of probability, which
can be called \emph{probabilistic} \cite{Franklin,DeMicheli1}.
\end{itemize}
In Section \ref{se:review} we shall briefly review both these methods, advancing some remarks, in particular about the
standard \emph{deterministic} regularization.

In this paper we approach the problem from a new viewpoint. The problem of reconstructing
the object from the image is regarded as a \emph{communication channel problem}. In this context
we estimate the messages which can be conveyed back from the data set
(the image) to reconstruct the signal (the object). The maximum number of these messages is limited
by the noise affecting the image. One could expect that the maximum number of these messages
tends to infinity as the noise affecting the image tends to zero. In this way the theory
can provide a precise and quantitative dependence of the resolution on the noise.

One of the main purposes of this paper is therefore to connect the regularization methods to information theory.
In Section \ref{se:capacity} we shall develop a \emph{topological information theory}, which
can be derived without making use of the tools proper of the probability theory. It is
rather based on the concepts of $\varepsilon$--entropy and $\varepsilon$--capacity, introduced
by Kolmogorov \cite{Kolmogorov}. The main result which we obtain is summarized by the following formula
\beq
M_\varepsilon \simeq 2^{\,S\log(1/\varepsilon)},
\label{5.1}
\eeq
where $M_\varepsilon$ is the maximum number of $\varepsilon$--distinguishable messages which can
be conveyed back through the channel from the noisy data set (the image) to recover the object;
$S=\Omega X_0/\pi$ is the \emph{Shannon number} (introduced above), $\varepsilon$ is a bound
on the noise affecting the image, and $\log x$ stands here and throughout the paper for the logarithm of $x$ to the base $2$.
In Section \ref{se:comparison}, instead, we develop an approach based on the
\emph{probabilistic information theory}, that is, the information theory which follows from the use of
\emph{probabilistic} methods \cite{Gabor}.
This allows us to compare the results obtained by the topological and the probabilistic information theory;
in particular, we can interpret the bounds on the information content of the image in terms of spectral distribution
of the noise and of the object.

\section{Review and remarks on regularization methods}
\label{se:review}
Equation \eqref{4.1} is a Fredholm equation of the first kind, and the operator $A$ is acting as follows:
$A:X\to Y$, where $X$ and $Y$ are the solution and the data space, respectively. We take here, for simplicity
and without loss of generality, $X=Y=L^2(-X_0/2,X_0/2)$. As we said in the Introduction,
the operator $A$ is self--adjoint, non--negative, and compact. Moreover, the unique solution of the
equation $Af=0$ is $f=0$. Then, we can say that the integral operator $A$ admits a complete set of orthogonal
eigenfunctions $\{\psi_k\}_{k=0}^\infty$ corresponding to a countably infinite set of real positive
eigenvalues $\lambda_0>\lambda_1>\lambda_2>\cdots$; moreover, $\lim_{k\to+\infty}\lambda_k=0$.
The properties of this integral operator have been already studied by several authors \cite{Slepian1,Slepian2,Slepian3,Frieden},
and the literature on this topic is quite extensive. Suppose that $S=\Omega X_0/\pi$ is sufficiently
large, then the eigenvalues $\lambda_k$ form a decreasing sequence $1>\lambda_0>\lambda_1> \cdots>0$, which enjoys
a step--like behavior, i.e., they are approximately equal to 1 for $k\lesssim S$, and then fall off to
zero exponentially (see Fig. \ref{fig:2} and Refs. \cite{Toraldo,Frieden}).
Since $A:X\to Y$ is compact then the range $\rmR(A)$ is not closed in the data space $Y$. Therefore, given a
data function $g\in Y$, it does not necessarily follow that there exists a solution $f\in X$.
Moreover, even if two data functions $g_1$ and $g_2$ belong to $\rmR(A)$ and their distance in $Y$ is
small, nevertheless the distance between $A^{-1}g_1$ and $A^{-1}g_2$ can be unlimited large, in view of the fact
that the inverse of the compact operator $A$ is not bounded ($X$ and $Y$ being infinite dimensional spaces).

\begin{figure}[b] 
\begin{center}
\leavevmode
\includegraphics[bb=30 0 387 336,scale=0.6]{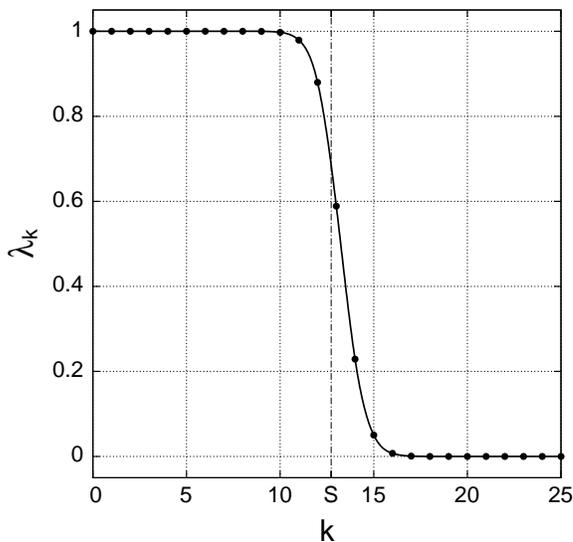}
\caption{\label{fig:2} The eigenvalues $\lambda_k$ (filled dots) of the kernel in \protect\eqref{3.1} with Shannon number $S=12.7$.}
\end{center}
\end{figure}

Since there always exists some inherent noise in the data, instead of \eqref{4.1} we have to deal with the following
equation
\beq
Af + n = \og \qquad (\og = g+n),
\label{1.2}
\eeq
where $n$ denotes the noise. Here we have assumed a purely additive model of noise, and hereafter we suppose that $n$ is a \emph{small}
perturbation of the data function, in order to still have $\og\in \rmR(A)$.

\subsection{Deterministic regularization methods}
\label{subse:deterministic}
Several methods of regularization have been proposed (see \cite{Tikhonov,Groetsch} and references quoted therein);
all of them aim at modifying one of the elements of the triplet $\{A,X,Y\}$, where $A$ is the integral operator
defined by \eqref{4.1}, and $X$ and $Y$ are the solution and data space, respectively (here we continue to assume
$X=Y=L^2(-X_0/2,X_0/2)$). Among these methods the procedure which is probably the most popular
consists in looking for the solution in a compact subset of the solution space $X$; then continuity
of the inverse operator follows from compactness. This restriction of the solution space, which ultimately leads to
a compact subset of $X$, is realized by means of suitable \emph{a priori} bounds that should represent some prior
knowledge on the solution. More precisely, in addition to the inequality
\beq
\|Af-\og\|_Y \leqslant \varepsilon \qquad (\varepsilon = \mathrm{constant}),
\label{2.2}
\eeq
which corresponds to a bound on the noise,
one also assume an \emph{a priori} bound on the solution of the following form
\beq
\|Bf\|_{\cZ} \leqslant E \qquad (E = \mathrm{constant}),
\label{3.2}
\eeq
where $\cZ$ denotes the \emph{constraint space}, and $B$ is the \emph{constraint operator}.
From bounds \eqref{2.2} and \eqref{3.2} we are led to determine the minimum of the following functional
\beq
\Phi(f) = \|Af-\og\|^2_Y + \mu^2 \, \|Bf\|_{\cZ}^2 \qquad \left(\mu = \frac{\varepsilon}{E}\right),
\label{4.2}
\eeq
which can be proved to be a regularized solution \cite{Magnoli}. Let $A^*$ denote the adjoint operator of $A$.
We take as constraint operator $B$ a self--adjoint operator; moreover, we assume
that $B^* B$ and $A^* A$ commute (this assumption does not restrict significantly
the theory and the applications). The space $\cZ$ is then composed of those functions $f \in L^2(-X_0/2,X_0/2)$
such that $\|Bf\|_{\cZ}$ is finite, i.e.,
\beq
\|Bf\|_{\cZ} = \left(\sum_{k=0}^\infty \beta_k^2 |f_k|^2 \right)^{1/2} < E \qquad (E = \mathrm{constant}),
\label{5.2}
\eeq
where $f_k=(f,\psi_k)$ ($(\cdot,\cdot)$ denoting the scalar product in $L^2(-X_0/2,X_0/2)$),
$B^* Bf = \sum_{k=0}^\infty \beta_k^2 f_k \psi_k$, $\beta_k^2$ being the eigenvalues of $B^*B$
(i.e., $B^*B\psi_k = \beta_k^2\psi_k$). Moreover, we require that $\lim_{k\to\infty}\beta_k^2=+\infty$, in
order to guarantee that the subset of the solution space, which is composed of those functions satisfying \eqref{3.2},
is compact. Now the functional $\Phi(f)$ has a unique minimum, given by
\beq
\of = \frac{A^* \og}{A^*A+\left(\varepsilon/E\right)^2 B^*B},
\label{6.2}
\eeq
which, by expanding $\og$ in terms of the functions $\psi_k$, can be written as follows,
\beq
\of = \sum_{k=0}^\infty \frac{\lambda_k\,\og_k}{\lambda_k^2+\left(\varepsilon/E\right)^2 \beta_k^2}\psi_k.
\label{7.2}
\eeq
Then, the following propositions can be proved. The proofs are given in very detailed form in \cite{Magnoli},
which refers to a different physical problem (the antenna synthesis). Nevertheless, the eigenfunctions
used there are the prolate spheroidal wave functions (as in the present problem), and the deterministic
regularization methods are given in variational form, which is appropriate for our case here.

\begin{proposition}
\label{pro:1}
For any function $f$ satisfying the bounds \eqref{2.2} and \eqref{3.2}, the following limit holds
\beq
\lim_{\varepsilon\to 0} \|f-\of\|_X=0 \qquad (E = \mathrm{fixed}).
\label{8.2}
\eeq
\end{proposition}

\begin{proof}
See Proposition 12 of  \cite{Magnoli}.
\end{proof}

In actual numerical computation it is often convenient to use truncated approximations. For instance,
the solution \eqref{7.2} leads to define the following approximation
\beq
\of^{\,(1)} \doteq \sum_{k=0}^{k_\beta} \frac{\og_k}{\lambda_k} \psi_k,
\label{9.2}
\eeq
where $k_\beta$ is the largest integer such that
\beq
\lambda_k \geqslant |\beta_k| \, \frac{\varepsilon}{E}.
\label{10.2}
\eeq

\begin{proposition}
\label{pro:2}
For any function $f$ satisfying bound \eqref{2.2}, the following limit holds
\beq
\lim_{\varepsilon\to 0} \|f-\of^{\,(1)}\| = 0 \qquad (E = \mathrm{fixed}).
\label{11.2}
\eeq
\end{proposition}

\begin{proof}
See Proposition 12 and its Corollary in \cite{Magnoli}.
\end{proof}

In several problems a weaker \emph{a priori} bound on the solution can be used by setting $B=\I$,
the identity operator. Therefore, instead of bound \eqref{3.2}, we have
\beq
\|Bf\|_{\cZ} \equiv \|f\|_{L^2(-X_0/2,X_0/2)} =
\left(\sum_{k=0}^\infty |f_k|^2\right)^{1/2} \leqslant E \qquad (E = \mathrm{constant}).
\label{12.2}
\eeq
In this case the unique minimum of functional \eqref{4.2} is given by
\beq
\of^{\,(2)} = \sum_{k=0}^{\infty} \frac{\lambda_k\og_k}{\lambda_k^2 + \left(\varepsilon/E\right)^2} \, \psi_k,
\label{13.2}
\eeq
and, accordingly, the following truncated approximation can be introduced
\beq
\of^{\,(3)} \doteq \sum_{k=0}^{k_\I} \frac{\og_k}{\lambda_k} \, \psi_k,
\label{14.2}
\eeq
where $k_\I$ is the largest integer such that
\beq
\lambda_k \geqslant \frac{\varepsilon}{E}.
\label{15.2}
\eeq
Both $\of^{\,(2)}$ and $\of^{\,(3)}$ converge to $f$ in the \emph{weak} sense. In fact, the following propositions can be proved.

\begin{proposition}
\label{pro:3}
For any function $f$ satisfying bounds \eqref{2.2} and \eqref{12.2}, the following limit holds
\beq
\lim_{\varepsilon\to 0} |([f-\of^{\,(2)}],v)| = 0
\qquad \left(\forall v \in L^2\left(-\frac{X_0}{2},\frac{X_0}{2}\right); E = \mathrm{fixed}\right).
\label{16.2}
\eeq
\end{proposition}

\begin{proof}
See Proposition 13 and its Corollary in \cite{Magnoli}.
\end{proof}

\begin{proposition}
\label{pro:3bis}
For any function $f$ satisfying bound \eqref{2.2} and \eqref{12.2}, the following limit holds
\beq
\lim_{\varepsilon\to 0} |([f-\of^{\,(3)}],v)| = 0
\qquad \left(\forall v \in L^2\left(-\frac{X_0}{2},\frac{X_0}{2}\right); E = \mathrm{fixed}\right).
\label{17.2}
\eeq
\end{proposition}

\begin{proof}
See Proposition 14 and its Corollary in \cite{Magnoli}.
\end{proof}

\begin{remark}
These regularization methods are not free from faults. We restrict ourselves to mention just two of them. For reason of simplicity
we shall focus on the approximation $\of^{\,(3)}$, but the same considerations hold also for $\of^{\,(1)}$. \\
(i) Approximation \eqref{14.2} is based on the truncation criterion \eqref{15.2}. Put, for simplicity and without loss of generality,
$E=1$. Then formula \eqref{15.2} reads: $\lambda_k \geqslant \varepsilon$. This means that the values of $\lambda_k$
(i.e., the eigenvalues of the operator $A$ representing the optical instrument) should be compared with the bound on the noise $\varepsilon$.
But this approach appears quite unnatural from the viewpoint of the experimental or physical sciences,
whose methodology rather suggests to compare the signal with the noise. In other words, the expansions should rather be truncated
at the value $k_\mathrm{p}$ of $k$ such that for $k>k_\mathrm{p}$ the Fourier coefficients $g_k=(g,\psi_k)$ of the noiseless data are smaller
or, at most, of the same order of magnitude of $\varepsilon$. In this case, in fact, it would be impossible to extract
information from the corresponding noisy coefficients $\og_k=(\og,\psi_k)$. \\
(ii) This second remark is strictly connected to the first one. It is easy to exhibit examples
of objects $f$ whose corresponding images $g$ have Fourier components small for low values of $k$, while the significant
contributions are carried by those Fourier components which are suppressed by condition \eqref{15.2} (i.e.,
$\lambda_k<\varepsilon$, $E=1$). This remark holds also for more refined solutions of the form \eqref{6.2},
which correspond to the minimization of functional \eqref{4.2}. Indeed, the minimization of this functional
works as a low--pass filter, whose action is smoothing the Fourier components $\og_k$ for high values of $k$.
This latter statement follows by noting that Proposition \ref{pro:1} holds if and only if
$\lim_{k\to\infty}\beta_k^2 = +\infty$. In conclusion, it is possible to give examples
where the standard deterministic regularization methods fail in spite of their rigorous mathematical
correctness, since these procedures do not guarantee that the bulk of the signal (of the object,
in our case) has been really recovered (see \cite{DeMicheli1}).
\label{rem:1}
\end{remark}

\subsection{Probabilistic regularization methods}
\label{subse:probabilistic}

We now want to reconsider Eq. \eqref{1.2} from a probabilistic point of view. With this in mind,
we re--write \eqref{1.2} in the following form
\beq
A\xi + \zeta = \eta,
\label{18.2}
\eeq
where $\xi$, $\zeta$, and $\eta$, which correspond to $f$, $n$, and $\og$ respectively, are Gaussian
weak random variables (w.r.v.) in the Hilbert space $L^2(-X_0/2,X_0/2)$ \cite{Balakrishnan}.
A Gaussian w.r.v. is uniquely defined by its mean element and its covariant operator. In the present case we denote
by $R_{\xi\xi}$, $R_{\zeta\zeta}$, and $R_{\eta\eta}$ the covariance operators of $\xi$, $\zeta$, and $\eta$ respectively.
Next, we make the following assumptions:

\begin{itemize}
\item[(I)] $\xi$ and $\zeta$ have zero mean, i.e., $m_\xi = m_\zeta = 0$;
\item[(II)] $\xi$ and $\zeta$ are uncorrelated, i.e., $R_{\xi\zeta}=0$;
\item[(III)] $R^{-1}_{\zeta\zeta}$ exists.
\end{itemize}

The third assumption is the mathematical formulation of the fact that all the components of the data
function are affected by noise. As it has been proved by Franklin (see formula (3.11) of \cite{Franklin}),
if the signal and the noise satisfy assumptions (I) and (II), then
\beq
R_{\eta\eta} = A R_{\xi\xi} A^* + R_{\zeta\zeta},
\label{19.2}
\eeq
and the cross--covariance operator is given by
\beq
R_{\xi\eta} = R_{\xi\xi} A^*.
\label{20.2}
\eeq
We also assume that $R_{\zeta\zeta}$ depends on a parameter $\varepsilon$ which tends to zero when
the noise vanishes, i.e., we write
\beq
R_{\zeta\zeta} = \varepsilon^2 \, \HH,
\label{21.2}
\eeq
where $\HH$ is a given operator (e.g., $\HH = \I = \mathrm{the ~ identity ~operator}$, in the case of white noise).
We can now state the following problem.

\begin{problem}
Given a value $\og$ of the w.r.v. $\eta$, find an estimate of the w.r.v. $\xi$.
\end{problem}

We first turn Eq. \eqref{18.2} into an infinite sequence of unidimensional equations by means of the orthogonal
projections,
\beq
\lambda_k \xi_k + \zeta_k = \eta_k \qquad (k=0,1,2,\ldots),
\label{22.2}
\eeq
where $\xi_k = (\xi,\psi_k)$, $\zeta_k = (\zeta,\psi_k)$, $\eta_k = (\eta,\psi_k)$ are Gaussian random variables.
Equations \eqref{22.2} can be obtained by formal expansions of the w.r.v. $\xi$, $\zeta$, and $\eta$ on the orthonormal
basis $\{\psi_k\}_{k=0}^\infty$ (which are the eigenfunctions of the operator $A$), i.e.,
$\xi=\sum_{k=0}^\infty\xi_k\psi_k$, $\zeta=\sum_{k=0}^\infty\zeta_k\psi_k$, and $\eta=\sum_{k=0}^\infty\eta_k\psi_k$.
Then, from \eqref{18.2} we obtain an infinite sequence of equalities of the following form: $(\lambda_k\xi_k+\zeta_k-\eta_k)\psi_k=0$
$(k=0,1,2,\ldots)$ from which Eqs. \eqref{22.2} follow. Let us remark that the expansions of $\xi$, $\zeta$, and $\eta$ are not
orthogonal expansions, since their coefficients are statistically interconnected or, in other words,
$\rmE\{\xi_m,\xi_n\}$ (and similarly $\rmE\{\zeta_m,\zeta_n\}$ and $\rmE\{\eta_m,\eta_n\}$; $\rmE\{\cdot\}$ denoting the expectation value)
does not in general vanish. It amounts to say that the coefficients of these expansions are not statistically
independent. Let us indeed remind that, in general, it is not possible to expand the process $\xi$ (or $\zeta$, or $\eta$)
in an orthogonal series on a finite interval, except in the limiting situation of a stationary white noise process.
This remark is relevant below in connection with the evaluation in the information theory approach.

Next, we can introduce the variances: $\rho^2_k=(R_{\xi\xi}\psi_k,\psi_k)$, $\varepsilon^2\nu_k^2=(R_{\zeta\zeta}\psi_k,\psi_k)$,
$\lambda_k^2\rho_k^2+\varepsilon^2\nu_k^2=(R_{\eta\eta}\psi_k,\psi_k)$. In view of assumptions (I) and (II), the probability
densities for $\xi_k$ and $\zeta_k$ can be written as follows
\beq
p_{\xi_k}(x)=\frac{1}{\sqrt{2\pi}\rho_k}\exp\left(-\frac{x^2}{2\rho_k^2}\right) \qquad (k=0,1,2,\ldots),
\label{23.2}
\eeq
and
\beq
p_{\zeta_k}(x)=\frac{1}{\sqrt{2\pi}\varepsilon\nu_k}\exp\left(-\frac{x^2}{2\varepsilon^2\nu_k^2}\right) \qquad (k=0,1,2,\ldots).
\label{24.2}
\eeq
By the use of Eqs. \eqref{22.2} we can also introduce the conditional probability density
$p_{\eta_k}(y | x)$ of the random variable $\eta_k$ for fixed $\xi_k=x$, which reads
\beq
p_{\eta_k}(y|x)=\frac{1}{\sqrt{2\pi}\,\varepsilon\nu_k}\exp\left[-\frac{(y-\lambda_k x)^2}{2\varepsilon^2\nu_k^2}\right]
=\frac{1}{\sqrt{2\pi}\,\varepsilon\nu_k}\exp\left[-\frac{\lambda_k^2}{2\varepsilon^2\nu_k^2}\left(x-\frac{y}{\lambda_k}\right)^2\right].
\label{25.2}
\eeq
Now, let us apply the Bayes formula that provides the conditional probability density of $\xi_k$ given $\eta_k$
through the following expression \cite{Middleton}
\beq
p_{\xi_k}(x|y)=\frac{p_{\xi_k}(x) \, p_{\eta_k}(y|x)}{p_{\eta_k}(y)},
\label{26.2}
\eeq
provided $p_{\eta_k}(y) \neq 0$. \\
Thus, if a realization of the random variable $\eta_k$ is given by $\og_k$, formula \eqref{26.2} becomes
\beq
p_{\xi_k}(x|\og)=A_k\exp\left(-\frac{x^2}{2\rho_k^2}\right)
\exp\left[-\frac{\lambda_k^2}{2\varepsilon^2\nu_k^2}\left(x-\frac{\og_k}{\lambda_k}\right)^2\right].
\label{27.2}
\eeq
Next, we introduce the following sets:
\begin{align}
\cI &= \{k \in\N \,:\,\lambda_k\rho_k \geqslant\varepsilon\nu_k\}, \label{IN.a} \\
\cN &= \{k \in\N \,:\,\lambda_k\rho_k <\varepsilon\nu_k\}. \label{IN.b}
\end{align}
We can now see that the conditional probability density \eqref{27.2} can be regarded as the product
of two Gaussian probability densities:
$p_1(x)=A_k^{(1)}\exp(-\frac{x^2}{2\rho_k^2})$ and
$p_2(x)=A_k^{(2)}\exp[-\frac{\lambda_k^2}{2\varepsilon^2\nu_k^2}(x-\frac{\og_k}{\lambda_k})^2]$ with
$A_k = A_k^{(1)}A_k^{(2)}$, whose variances are respectively given by $\rho_k^2$ and $(\varepsilon\nu_k/\lambda_k)^2$.
Now, if $k\in\cI$ the variance associated with $p_2(x)$ is smaller than the corresponding variance of $p_1(x)$,
and vice versa if $k\in\cN$. Therefore, it appears reasonable to consider as an acceptable approximation of
$\langle\xi_k\rangle$, i.e. the mean value of the random variable $\xi_k$, the mean
value associated with the density $p_2(x)$ if $k\in\cI$, and the mean value associated with the density
$p_1(x)$ if $k\in\cN$. We can then write the following approximation
\beq
\langle\xi_k\rangle =
\begin{cases}
\frac{\sds\og_k}{\sds\lambda_k} & \quad\text{if $k \in \cI$}, \label{29.a.2}  \\
~0 & \quad\text{if $k \in \cN$}.
\end{cases}
\eeq
Consequently, given the value $\og$ of the w.r.v. $\eta$, we are led to consider the following linear estimator
of $\xi$
\beq
T\eta = \sum_{k\in\cI} \frac{\og_k}{\lambda_k}\,\psi_k.
\label{30.2}
\eeq
In order to pass from heuristic considerations to rigorous statements, we must prove that the linear
estimator \eqref{30.2} leads to a probabilistically regularized solution. For this purpose, we must evaluate
the global mean--squared error associated with the linear estimator \eqref{30.2}, i.e.
$\rmE\{\|\xi-T\eta\|^2\}$, along with
$\rmE\{\|\xi\|^2\}=\sum_{k=0}^\infty(R_{\xi\xi}\psi_k,\psi_k)=\mathrm{Trace}\,(R_{\xi\xi})$. \\
The following propositions can be proved.

\begin{proposition}
\label{pro:4}
(i) If $\lim_{k\to\infty}(\lambda_k\rho_k/\nu_k)=0$, then the set $\cI$ is finite for any $\varepsilon>0$. \\
(ii) Assuming that the limit stated in (i) holds and, in addition, that $R_{\xi\xi}$ is an operator of trace class, then the
following relationship holds
\beq
\rmE\left\{\|\xi-T\eta\|^2\right\}=
\sum_{k\in\cN}\rho_k^2 + \sum_{k\in\cI} \frac{\varepsilon^2\nu_k^2}{\lambda_k^2} < \infty.
\label{31.2}
\eeq
\end{proposition}

\begin{proof}
See Proposition 3.3 of  \cite{DeMicheli1}.
\end{proof}

\begin{proposition}
\label{pro:5}
If the covariance operator $R_{\xi\xi}$ is of trace class, and if the set $\cI$ is finite (see Proposition \ref{pro:4}),
then the following limit holds
\beq
\lim_{\varepsilon\to 0} ~ \rmE\left\{\|\xi-T\eta\|^2\right\}=0,
\label{32.2}
\eeq
i.e., the linear estimator of $\xi$ given by formula \eqref{30.2} gives a probabilistically regularized solution to Problem \ref{pro:1}.
\end{proposition}

\begin{proof}
See Proposition 3.5 of \cite{DeMicheli1}.
\end{proof}

\begin{remark}
As we have already remarked above, the deterministic regularization methods do not guarantee that the
bulk of the signal (the object, in our case) has been really recovered. Conversely, the probabilistic
regularization methods (i.e., the solution given by formula \eqref{30.2}) can really reconstruct,
within a certain degree of approximation, the bulk of the object, once the sets $\cI$ and $\cN$ have
been neatly separated. In fact, as we shall see in Section \ref{se:comparison}, the Gaussian
random variables $\eta_k$, associated with the set $\cI$, contain a significant amount of information
on the corresponding variables $\xi_k$, whereas in the random variables $\eta_k$,
associated with the set $\cN$, the noise is prevailing. At this point the problem is how
to split the set of the Gaussian random variables $\{\eta_k\}$ into the two sets $\cI$ and $\cN$.
This task can be achieved by computing the correlation function of the random variables $\eta_k$, which are
the probabilistic counterpart of the coefficients $\og_k$. Let us indeed recall that the coefficients
$\eta_k=(\eta,\psi_k)$, obtained by the formal expansion $\eta=\sum_{k=0}^\infty\eta_k\psi_k$ are not
statistically independent. These statistical methods require great caution and involve delicate
mathematical questions, which have been studied in \cite{DeMicheli1,DeMicheli2}, and we do not
return on these problems here. In particular, in \cite{DeMicheli1} some explicit examples have been
shown, where the deterministic regularization method fails, whereas the statistical one can actually reconstruct
the solution of the integral equation considered. Analogous statistical methods have been also used
in optics \cite{DeMicheli2}, in connection with the object restoration in the case of spatially incoherent illumination.
In this case, in particular, a positivity constraint has been incorporated into the probabilistically
regularized solution by means of a quadratic programming technique. Several examples were shown,
and satisfactory results had been obtained. \\
As a final remark, we point out that the deterministic and the statistical regularization methods can be used jointly in the
following sense. Assuming that we know \emph{a priori} global bounds on the solution such that deterministic
regularized solutions can be tried, then their reliability can be tested by using statistical correlation
methods along the lines suggested by the probabilistic regularization procedures.
\label{rem:2}
\end{remark}

\section{Topological information theory: the $\boldsymbol{\varepsilon}$--entropy of the image}
\label{se:capacity}

Let us return to the deterministic regularization methods, and to the related
\emph{a priori} truncation criteria. Consider bound \eqref{15.2} where, for simplicity and without loss of generality,
we put $E=1$. Accordingly, we consider the approximation $\of^{\,(3)}=\sum_{k=0}^{k_\I}(\og_k/\lambda_k)\psi_k$,
where $k_\I$ is the largest integer such that $\lambda_k\geqslant\varepsilon$. As proved in
Proposition \ref{pro:3}, $\of^{\,(3)}$ converges \emph{weakly} to $f$ and, consequently, only a \emph{weak} continuity
can be guaranteed in the restored solution. \\
We now make two additional assumptions:

\begin{enumerate}
\item We assume that the noise $n$ is moderate enough, namely,
it is such that the noisy image belongs to the range of $A$: $\og\in\rmR(A)$.
\item We assume that $k_\I \simeq k_\mathrm{p}$, i.e. the truncation number $k_\I$ associated with the approximation
$\of^{\,(3)}$ is very close to $k_\mathrm{p}$, which is the value of $k$ such that for $k>k_\mathrm{p}$ the Fourier components
$g_k=(g,\psi_k)$ of the noiseless data are smaller or, at most, of the same order of magnitude of $\varepsilon$
(see Remark \ref{rem:1}). It is obvious that in making this assumption we suppose that
the modulus of the coefficients $g_k$ decreases for increasing values of $k$. If these assumptions are true,
then we can exclude those ``\emph{pathological}'' examples in which the bulk of the object is not recovered by the
approximation $\of^{\,(3)}$.
\end{enumerate}

In view of the \emph{a priori} bound \eqref{12.2} with $E=1$, we are led to consider the unit ball in the solution space
$X \equiv L^2(-X_0/2,X_0/2)$: i.e., the set $\{f\in X\,:\,\|f\|_X\leqslant 1\}$; the operator $A$ maps the unit ball onto a compact ellipsoid
$\cE \in \mathrm{Range}\,(A)$, contained in the data space $Y \equiv L^2(-X_0/2,X_0/2)$, whose semi--axes lengths
are the eigenvalues $\lambda_k$ of the operator $A$.

Let us now recall some basic definitions from the information theory \cite{Kolmogorov}:
\begin{itemize}
\item[(a)] In the theory of information, the unit of a \emph{collection of information} is the
amount of information in one binary sign (that is, designating whether it is 0 or 1).
\item[(b)] The \emph{entropy} of a collection of possible \emph{communications}, undergoing
transmission with a specified accuracy, is defined as the number of binary signs necessary
to transmit an arbitrary one of these communications with a given accuracy.
\item[(c)] The \emph{capacity} of a transmitting apparatus is defined as the number of binary signs that it
can transmit reliably.
\end{itemize}
Coming back to the compact ellipsoid $\cE$, we recall some basic definitions which give a numerical
estimate of its \emph{massiveness} \cite{Kolmogorov,Lorentz}:
\begin{itemize}
\item[(a')] A family $Y_0,\ldots,Y_n$ of subsets of $Y$ is an $\varepsilon$--covering of $\cE$ if the
diameter of each $Y_k$ does not exceed $2\varepsilon$ and if the sets $Y_k$ cover $\cE$: i.e.,
$\cE \subset\cup_{k=0}^n Y_k$.
\item[(b')] Points $y_0,\ldots,y_m$ of $\cE$ are called $\varepsilon$--distinguishable if the distance
between each two of them exceeds $\varepsilon$.
\end{itemize}
Since $\cE$ is compact, then a finite $\varepsilon$--covering exists for each $\varepsilon>0$, and, moreover,
$\cE$ can contain only finite many $\varepsilon$--distinguishable points. For a given $\varepsilon>0$,
the number of sets $Y_k$ in a covering family depends on the family, but the minimal value of $n$, $N_\varepsilon(\cE)\doteq\min n$,
is an invariant of the set $\cE$, which depends only on $\varepsilon$. Its logarithm, that is, the function
$H_\varepsilon(\cE)\doteq\log N_\varepsilon(\cE)$ is the $\varepsilon$--entropy of the set $\cE$,
and gives the length of the binary sequence from which a signal in $\cE$ can be reconstructed up to $\varepsilon$ accuracy.
Analogously, the number $m$ in definition (b') above depends on the choice of the points, but its maximum
$M_\varepsilon(\cE)\doteq\max m$, is an invariant of the set $\cE$, and represents the maximum number of
$\varepsilon$--distinguishable messages that can be conveyed back in the backward channel to reconstruct the object:
i.e., the maximum number of those data which satisfy the inequalities:
$\|\og^{(i)}-\og^{(k)}\|_Y>\varepsilon$ for all $i\neq k$, $\og^{(i)},\og^{(k)}\in\cE$.
Its logarithm, that is, the function $C_\varepsilon(\cE)=\log M_\varepsilon(\cE)$, is the $\varepsilon$--capacity of the set $\cE$,
and provides the length (in binary units) of the messages that can be reliably transmitted
in the backward channel.

The following inequalities hold \cite{Kolmogorov,DeMicheli3}:
\beq
H_\varepsilon(\cE) \leqslant C_\varepsilon(\cE) \leqslant H_{\varepsilon/2}(\cE).
\label{1.3}
\eeq
Then, in order to obtain estimates for the $\varepsilon$--capacity $C_\varepsilon(\cE)$, our aim is now to look for a
lower bound for $H_\varepsilon(\cE)$ and an upper bound for $H_{\varepsilon/2}(\cE)$. For this purpose, let us consider
the finite dimensional subspace $Y_{k_\I}$ of $Y$, spanned by the first $k_\I+1$ axes of $\cE$, and put
$\cE_{k_\I}=\cE \cap Y_{k_\I}$. Then, $\cE_{k_\I}$ is a finite dimensional ellipsoid whose volume is just
$\prod_{k=0}^{k_\I}\lambda_k$ times the volume $\Omega_{k_\I}$ of the unit ball in $Y_{k_\I}$. Since the volume of
an $\varepsilon$--ball in $Y_{k_\I}$ is $\varepsilon^{(k_\I+1)}\Omega_{k_\I}$, we see that in order to cover the ellipsoid
$\cE$ by the $\varepsilon$--balls we shall need at least $\prod_{k=0}^{k_\I}(\lambda_k/\varepsilon)$ such balls.
From this it follows that \cite{Prosser,Gelfand1}:
\beq
\prod_{k=0}^{k_\I} \frac{\lambda_k}{\varepsilon} \leqslant N_\varepsilon(\cE),
\label{2.3}
\eeq
and, therefore, we have the following lower bound for the $\varepsilon$--entropy $H_\varepsilon(\cE)$:
\beq
\sum_{k=0}^{k_\I} \log\frac{\lambda_k}{\varepsilon} \leqslant \log N_\varepsilon(\cE) = H_\varepsilon(\cE).
\label{3.3}
\eeq
The determination of an upper bound for $H_{\varepsilon/2}(\cE)$ is more involved, and we limit ourselves to
report the result \cite{Prosser,Gelfand1}:
\beq
H_{\varepsilon/2}(\cE) \leqslant
k_\I\left(\frac{\varepsilon}{4}\right)
\left\{\log\left(\frac{1}{\varepsilon}\right)+\log 6 + \frac{1}{2}\log k_\I\left(\frac{\varepsilon}{4}\right)\right\},
\label{4.3}
\eeq
where $k_\I(\varepsilon/4)$ represents the number of terms in the sequence $\{\lambda_k\}_{k=0}^\infty$
which are larger or equal to $(\varepsilon/4)$.

Now, we come back to the optical problem, specifically to Eq. \eqref{3.1}, and investigate the behavior of the
$\varepsilon$--entropy $H_{\varepsilon}(\cE)$ in the limit of low level of noise.
Assuming that the Shannon number $S=\Omega X_0/\pi$ is sufficiently large, the eigenvalues $\lambda_k$ can be
approximated with 1 for $k\leqslant S$ (see, e.g., Fig. \ref{fig:2}), whereas, for $k>S$, the
eigenvalues $\lambda_k$ fall off to zero exponentially \cite{Bertero1}.
Consider now the bound in \eqref{3.3}; for $\varepsilon$ sufficiently small, we have $k_\I(\varepsilon) > S$, and the
sum in \eqref{3.3} can be split into two parts:
\beq
\sum_{k=0}^{k_\I} \log\frac{\lambda_k}{\varepsilon}
= \sum_{k=0}^{\lfloor S\rfloor-1} \log\frac{\lambda_k}{\varepsilon} + \sum_{k=\lfloor S \rfloor}^{k_\I} \log\frac{\lambda_k}{\varepsilon},
\label{3.3bis}
\eeq
where the symbol $\lfloor x \rfloor$ stands for the integral part of $x$.
Since for $k<S$ we have $\lambda_k \simeq 1$, the contribution of the first sum on the r.h.s. of \eqref{3.3bis}
is about $S\log(1/\varepsilon)$. Instead, for $k \geqslant S$ we have $\lambda_k \simeq \varepsilon$, so that the
second sum on the r.h.s. of \eqref{3.3bis} is nearly null. Then, from \eqref{4.3} we obtain the following lower bound
for the $\varepsilon$--entropy:
\beq
H_\varepsilon(\cE) \sim S\,\log\left(\frac{1}{\varepsilon}\right).
\label{3.3tris}
\eeq
Therefore, we can conclude that the maximum number of $\varepsilon$--distinguishable messages, which
can be conveyed back from the image to recover the object, at least should be:
\beq
M_\varepsilon(\cE) \gtrsim 2^{\,S\log(1/\varepsilon)}
\xrightarrow[\varepsilon\rightarrow 0] {} \infty.
\label{8.3}
\eeq

Next, we can consider formula \eqref{4.3}, which limits superiorly the number of
$\varepsilon$--distinguishable messages.
First we note that the eigenvalues $\lambda_k$ decrease exponentially for $k\to\infty$; precisely, we have \cite{Bertero1}:
$\lambda_k = \mathrm{O}(\exp[-2k\log(k/c)]/k)$, $c=\mathrm{constant}$. Then,
it follows that, for $\varepsilon\to 0$, $k_\I(\varepsilon/4) \sim \frac{1}{2}\log(1/\varepsilon)$,
and the leading term, for $\varepsilon\to 0$, on the r.h.s. of \eqref{4.3} is:
$k_\I(\varepsilon/4)\log(1/\varepsilon)$. We thus have:
\beq
H_{\varepsilon/2}(\cE)
\, \staccrel{\sds\sim}{\varepsilon\to 0} \,
k_\I\left(\frac{\varepsilon}{4}\right)\log\left(\frac{1}{\varepsilon}\right)
\sim\frac{1}{2}\log^2\left(\frac{1}{\varepsilon}\right).
\label{6.3}
\eeq

Summarizing, from \eqref{1.3}, \eqref{3.3tris}, and \eqref{6.3} we have, for $\varepsilon$ sufficiently small:
\beq
S\,\log\left(\frac{1}{\varepsilon}\right) \lesssim C_\varepsilon(\cE) \lesssim \frac{1}{2}\log^2\left(\frac{1}{\varepsilon}\right).
\label{111.3}
\eeq
These latter inequalities require: $S < \frac{1}{2}\log\left(\frac{1}{\varepsilon}\right)$, that is,
$\varepsilon<2^{-2S}$. In other words, this means that as long as the noise level is not too small, i.e. for $\varepsilon>2^{-2S}$,
the $\varepsilon$--capacity is essentially: $C_\varepsilon(\cE) \simeq S\,\log\left(1/\varepsilon\right)$
(to have a flavor of the numbers, for the operator $A$ whose eigenvalues are shown in Fig. \ref{fig:2}, with $S=12.7$, this approximation of the
the $\varepsilon$--capacity holds for $\varepsilon \gtrsim 10^{-7.6}$ or, equivalently, for a
signal--to--noise ratio: $(E/\varepsilon) \lesssim 76 \,\mathrm{dB}$).
Instead, when the noise gets smaller, i.e. for $\varepsilon<2^{-2S}$, the $\varepsilon$--capacity may increase faster when $\varepsilon\to 0$,
remaining (approximately) within the range specified by inequalities \eqref{111.3}.

\section{Comparing probabilistic and topological information theory}
\label{se:comparison}

Let us return now to the probabilistic regularization methods, and evaluate the amount of information
on the random variable $\xi_k$, which is contained in the random variable $\eta_k$; we have \cite{Gelfand2}:
\beq
J\left(\xi_k,\eta_k\right)=-\frac{1}{2}\ln(1-r_k^2) \qquad (k=0,1,2,\ldots),
\label{1.4}
\eeq
($\ln x$ denotes the logarithm of $x$ to the base $e$), where $r_k$ is given by:
\beq
r_k^2 = \frac{\left|\rmE\left\{\xi_k,\eta_k^*\right\}\right|^2}{\rmE\left\{|\xi_k|^2\right\} \rmE\left\{|\eta_k|^2\right\}}
= \frac{(\lambda_k\rho_k)^2}{(\lambda_k\rho_k)^2 + (\varepsilon\nu_k)^2}  \qquad (k=0,1,2,\ldots),
\label{2.4}
\eeq
and the equality $R_{\xi\eta}=R_{\xi\xi}A^*$ (see \eqref{20.2}) has been used. From \eqref{1.4} and \eqref{2.4} it follows:
\beq
J\left(\xi_k,\eta_k\right)=\frac{1}{2}\ln\left(1+\frac{\lambda_k^2\rho_k^2}{\varepsilon^2\nu_k^2}\right)
 \qquad (k=0,1,2,\ldots).
\label{3.4}
\eeq
Let us now consider the sets $\cI$ and $\cN$, defined in \eqref{IN.a} and \eqref{IN.b}.
We see that, for the random variables $\xi_k$ and $\eta_k$
whose $k$--values belong to the set $\cN$, Eq. \eqref{3.4} gives:
\beq
J\left(\xi_k,\eta_k\right) < \frac{1}{2} \ln 2 \qquad (k\in\cN).
\label{4.4}
\eeq
We can thus say that, in the components $\eta_k$ whose values of $k$ belong to the set $\cN$
(for simplicity we write $\eta_k\in\cN$), the noise is prevailing and therefore they can be neglected in the approximate
reconstruction of the object, in agreement with formula \eqref{29.a.2}. \\
Conversely, the components $\eta_k\in\cI$
contain a significant amount of information on the corresponding components $\xi_k$. We can thus write, with obvious notation:
\beq
\oJ \doteq \sum_{k\in\cI}J\left(\xi_k,\eta_k\right)=
\sum_{k\in\cI}\ln\sqrt{1+\frac{\lambda_k^2\rho_k^2}{\varepsilon^2\nu_k^2}}.
\label{5.4}
\eeq

\begin{remark}
The quantity $\oJ$ in \eqref{5.4} is not the total information $J(\xi,\eta)$. In fact, the pairs
$\{\xi_i,\eta_j\}$ ($i\neq j$) are not mutually independent. A linear coordinate transformation
could always been chosen in such a way that all the components 
$\{\xi,\eta\}=\{\txi_0,\txi_1,\ldots,\txi_k;\teta_0,\teta_1,\ldots,\teta_k\}$
(with the exception of the pairs $\{\txi_j,\teta_j\}$, ($j=0,1,2,\ldots,k$)) are mutually independent.
But this would imply to introduce a basis $\{\tpsi_k\}_{k=0}^\infty$, which differs from that obtained
by the eigenfunctions $\{\psi_k\}_{k=0}^\infty$ of the operator $A$ that we used in the derivation of
the probabilistic regularization methods. Therefore, we limit ourselves to evaluate $\sum_{k\in\cI}J(\xi_k,\eta_k)$,
which does not provide the total amount of information $J(\xi,\eta)$ but represents only an approximation of it.
\label{rem:3}
\end{remark}

Next, we make the following approximation:
\beq
\oJ =
\sum_{k\in\cI} \ln\sqrt{1+\frac{\lambda_k^2\rho_k^2}{\varepsilon^2\nu_k^2}} \simeq
\sum_{k\in\cI} \ln \left|\frac{\lambda_k\rho_k}{\varepsilon\nu_k}\right|,
\label{6.4}
\eeq
which is admissible if $\lambda_k\rho_k\geqslant\varepsilon\nu_k$: i.e., for the components $\eta_k\in\cI$.
We now assume that: $\rho_k\sim\nu_k$ for $k\in\cI$. Then, from \eqref{6.4} we obtain:
\beq
\oJ = \sum_{k\in\cI} J(\xi_k,\eta_k) \simeq \sum_{k\in\cI} \ln\frac{\lambda_k}{\varepsilon}.
\label{7.4}
\eeq
In particular, let us note that from the assumption $\rho_k\sim\nu_k$ (for $k\in\cI$) it follows that the set
$\cI$ is composed of those components such that $\lambda_k\geqslant\varepsilon$, which is precisely
the truncation criterion \eqref{15.2} (with $E=1$) which generates the approximation $\of^{\,(3)}$.
Thus, from \eqref{7.4} we have:
\beq
\sum_{k\in\cI} J(\xi_k,\eta_k) \simeq \sum_{k\in\cI} \ln\frac{\lambda_k}{\varepsilon} =
\sum_{k=0}^{k_\I}\ln \frac{\lambda_k}{\varepsilon},
\label{8.4}
\eeq
which coincides with the lower bound on the $\varepsilon$--capacity (see Eq. \eqref{3.3}) up to an
immaterial conversion factor between logarithms to different bases. Again, as we made for obtaining formula
\eqref{3.3tris}, we have $\lambda_k\simeq 1$ for $k\leqslant S$, which finally yields:
\beq
\oJ=\sum_{k\in\cI} J(\xi_k,\eta_k) \simeq
S\,\ln\left(\frac{1}{\varepsilon}\right).
\label{9.4}
\eeq
Correspondingly, the maximum number of $\varepsilon$--distinguishable messages which can be conveyed back
in the backward channel from the image to recover the object, can therefore written as (neglecting the
conversion factor between $\log x$ and $\ln x$):
\beq
M_\varepsilon(\cE) =
2^{\,C_\varepsilon(\cE)} \simeq
2^{\,S\log(1/\varepsilon)} \simeq
2^\oJ =
2^{\,\left\{\sum_{k\in\cI}J(\xi_k,\eta_k)\right\}},
\label{12.4}
\eeq
which, as expected, tends to infinity as $\varepsilon$ tends to zero.

Returning to Eq. \eqref{6.4}, let us now make the following assumption: $\lambda_k\rho_k \sim \nu_k$ for $k\in\cI$. We
have:
\beq
\oJ \simeq \sum_{k\in\cI} \ln\left|\frac{\lambda_k\rho_k}{\varepsilon\nu_k}\right|
\simeq\sum_{k\in\cI}\ln\left(\frac{1}{\varepsilon}\right)
=k_\I(\varepsilon)\ln\left(\frac{1}{\varepsilon}\right).
\label{13.4}
\eeq
Now, recalling that the sequence of eigenvalues $\lambda_k$ falls off exponentially to zero for $k$ sufficiently large, from \eqref{13.4} we obtain:
\beq
\oJ \simeq
k_\I(\varepsilon)\ln\left(\frac{1}{\varepsilon}\right) \simeq \frac{1}{2}\ln^2\left(\frac{1}{\varepsilon}\right),
\label{14.4}
\eeq
which coincides with the upper bound on the $\varepsilon$--capacity given in \eqref{6.3}.

Summarizing, we see that for a given (small) level of noise $\varepsilon$,
the two extremal cases for the maximum number of $\varepsilon$--distinguishable data--messages which represent the information
that can be sent back through the backward channel to reconstruct the object, are related to the spectral distribution of the noise.
The lower limit is obtained when, for $k\in\cI$, the spectral distribution of the noise (i.e., $\nu_k$) \emph{coincides} with the distribution of the
object (i.e., $\rho_k$). The upper bound corresponds to the case when, for $k\in\cI$, the spectral distribution of the noise
\emph{coincides} with that of the image (i.e., $\lambda_k\rho_k$).

\section{Conclusions}
\label{se:conclusions}

Let us start from the classical Whittaker--Kotel'nikov--Shannon sampling theorem \cite{Jerri}, which states that a function,
whose Fourier transform vanishes outside a certain interval of length $2\Omega$, can be reconstructed by a discrete collection
of its values, chosen in arithmetic progression with difference $\pi/\Omega$. Since the image $g(y)$ is a band--limited
function, it could, in principle, be reconstructed by an infinite collection of its samples, taken at equidistant points
spaced $\pi/\Omega$ apart. More realistically, the image $g(y)$ can be reconstructed in an interval of length $X_0$ by a finite
collection $S=\Omega X_0/\pi$ of its samples. The classical Rayleigh resolution distance $R$ equals the
Nyquist distance $\pi/\Omega$, while the Shannon number $S$ turns out to be given by
$\mathrm{Trace}\,(A) = \sum_{k=0}^\infty \lambda_k$ \cite{Kato,Gori}.

Since both the image $g(y)$ and the Fourier transform of the object $F(\omega)$ are entire functions in the
complex variables $y$ and $\omega$ respectively, they can be analytically continued beyond the interval where they
are known. Consider, for instance, $F(\omega)$: in principle, it might be possible to extrapolate the function outside the data band
$[-\Omega,\Omega]$ by making use of appropriate regularization methods of ill--posed problems, and then to find
an estimate of it over a broader band, say, $[-W,W]$. This would imply a better resolution $\pi/W$:
this improvement can be called \emph{super--resolution}. In fact, it has been shown that whenever the Shannon number is not too large
(i.e., not much greater than unity) the behavior of the eigenvalues $\lambda_k$ is not similar to that of a step
function (see Fig. \ref{fig:2}), and therefore, the extrapolation of $F(\omega)$ out of band is indeed possible \cite{Bertero2}.

We have focused on aspects of the problem by analyzing the inverse imaging problem from two different viewpoints:
the classical information theory based on probabilistic methods,
and the Kolmogorov's $\varepsilon$--capacity (and entropy), which can be thought of as a form of information theory based on topological concepts.
The main results obtained, if a few conditions (specified at the beginning of Section \ref{se:capacity}) are satisfied,
can be summarized in the following points:
\begin{itemize}
\item[(a)] The $\varepsilon$--capacity of the image data set is essentially given by:
\beq
C_\varepsilon(\cE) \sim S\,\log\left(\frac{1}{\varepsilon}\right),
\label{1.5}
\eeq
where $S$ is the Shannon number. Consequently, the maximum number of $\varepsilon$--distinguishable messages which can be conveyed
back in the backward channel from the image to reconstruct the object is given by:
\beq
M_\varepsilon(\cE) \sim 2^{\,S\log(1/\varepsilon)}.
\label{2.5}
\eeq
\item[(b)] For $\varepsilon$ sufficiently small, i.e. $\varepsilon \lesssim 2^{-S}$, the $\varepsilon$--capacity is bounded above by:
\beq
C_\varepsilon(\cE) \lesssim \frac{1}{2}\log^2\left(\frac{1}{\varepsilon}\right) \xrightarrow[\varepsilon\rightarrow 0] {} \infty.
\label{2.5bis}
\eeq
\item[(c)] The upper and lower bounds on the information content of the noisy image (i.e., $C_\varepsilon(\cE)$)
obtained by the topological information theory may be interpreted
within the framework of the probabilistic information theory. In fact, the sum $\oJ$ of the information contained in the random
variables $\eta_k$, which represent the noisy image,
on the corresponding random variable $\xi_k$, which represent the object, is given by:
\begin{itemize}
\item[$(\mathrm{c}_1)$] If, for $k\in\cI$, the spectral distribution of the noise is as that of the object, i.e. $\nu_k \sim \rho_k$:
\begin{equation}
\oJ=\sum_{k\in\cI} J(\xi_k,\eta_k) \simeq S\,\ln\left(\frac{1}{\varepsilon}\right).
\label{3.5}
\eeq
\item[$(\mathrm{c}_2)$] If, for $k\in\cI$, the spectral distribution of the noise is as that of the image, i.e. $\nu_k \sim \lambda_k\rho_k$:
\beq
\oJ \simeq \frac{1}{2}\ln^2\left(\frac{1}{\varepsilon}\right).
\label{3.5bis}
\eeq
\end{itemize}
\item[(d)] The maximum number of $\varepsilon$--distinguishable messages which can be conveyed back from the
image to reconstruct the object is given by:
\beq
M_\varepsilon(\cE) =
2^{C_\varepsilon(\cE)} \sim
2^{\,\left\{\sum_{k\in\cI}J(\xi_k,\eta_k)\right\}}
\xrightarrow[\varepsilon\rightarrow 0] {} \infty.
\label{4.5}
\eeq
\end{itemize}


\begin{thebibliography}{99}

\bibitem{Wolter}
H. Wolter,
On basic analogies and principal differences between optical and electronic information,
in \emph{Progress in Optics}, Vol. I, E. Wolf, ed. (North Holland, 1961), pp. 155--210.

\bibitem{Toraldo}
G. Toraldo di Francia,
Degrees of freedom of an image,
J. Opt. Soc. Am. \textbf{59}, 799--804 (1969)
(see also the references quoted therein).

\bibitem{Hadamard}
J. Hadamard,
\emph{Lectures on the Cauchy Problem in Linear Partial Differential Equations}
(Yale University Press, 1923).

\bibitem{Viano1}
G.A. Viano,
On the extrapolation of optical image data,
J. Math. Phys. \textbf{17}, 1160--1165 (1976).

\bibitem{Bertero1}
M. Bertero, C. De Mol, and G.A. Viano,
The stability of inverse problems,
in \emph{Inverse Scattering Problems in Optics},
H.P. Baltes, ed., \emph{Topics in Current Physics}, Vol. 20,
(Springer, 1980), pp. 161--214.

\bibitem{Born}
M. Born and E. Wolf,
\emph{Principles of Optics}
(Pergamon, 1959), p. 166.

\bibitem{Stern}
A. Stern and B. Javidi,
Shannon number and information capacity of three--dimensional integral imaging,
J. Opt. Soc. Am. A \textbf{21}, 1602--1612 (2004).

\bibitem{Tikhonov}
A. Tikhonov and V. Arsenine,
\emph{M\'ethodes de R\'esolution de Probl\`emes Mal Poses}
(Mir, 1976).

\bibitem{Groetsch}
C.W. Groetsch,
\emph{The Theory of Tikhonov Regularization for Fredholm Integral Equations of the First Kind}
(Pitman, 1984).

\bibitem{Franklin}
J.M. Franklin,
Well--posed stochastic extensions of ill--posed linear problems,
J. Math. Anal. Appl. \textbf{31}, 682--716 (1970).

\bibitem{DeMicheli1}
E. De Micheli, N. Magnoli, and G.A. Viano,
On the regularization of Fredholm integral equations of the first kind,
SIAM J. Math. Anal. \textbf{29}, 855--877 (1998).

\bibitem{Kolmogorov}
A.N. Kolmogorov and V.M. Tihomirov,
$\varepsilon$--entropy and $\varepsilon$--capacity of sets in functional spaces,
Amer. Math. Soc. Trans. \textbf{17}, 277--364 (1961).

\bibitem{Gabor}
D. Gabor,
Light and information,
in \emph{Progress in Optics}, Vol. I, E. Wolf, ed. (North Holland, 1961), pp. 111--153.

\bibitem{Slepian1}
D. Slepian and H.O. Pollack,
Prolate spheroidal wave functions, Fourier analysis and uncertainty -- I,
Bell System Tech. J. \textbf{40}, 43--64 (1961).

\bibitem{Slepian2}
D. Slepian,
Prolate spheroidal wave functions, Fourier analysis and uncertainty -- IV: Extension to many dimensions;
Generalized prolate spheroidal wave functions,
Bell System Tech. J. \textbf{43}, 3009--3057 (1964).

\bibitem{Slepian3}
D. Slepian and E. Sonnenblick,
Eigenvalues associated with prolate spheroidal wave functions of zero order,
Bell System Tech. J. \textbf{44}, 1745--1759 (1965).

\bibitem{Frieden}
B.R. Frieden,
Evaluation, design, and extrapolation methods for optical signals, based on use of the prolate functions,
in \emph{Progress in Optics}, Vol. IX, E. Wolf, ed. (North Holland, 1972), pp. 311--407.

\bibitem{Magnoli}
N. Magnoli and G.A. Viano,
The source identification problem in electromagnetic theory,
J. Math. Phys. \textbf{38}, 2366--2388 (1997).

\bibitem{Balakrishnan}
A.V. Balakrishnan,
\emph{Applied Functional Analysis}
(Springer--Verlag, 1976), Chapter 6.

\bibitem{Middleton}
D. Middleton,
\emph{An Introduction to Statistical Communication Theory}
(McGraw--Hill, 1960), Chapter 6.

\bibitem{DeMicheli2}
E. De Micheli and G.A. Viano,
Probabilistic regularization in inverse optical imaging,
J. Opt. Soc. Am. A \textbf{17}, 1942--1951 (2000).

\bibitem{Lorentz}
G.G. Lorentz,
\emph{Approximation of Functions}
(Holt, Rinehart and Winston, 1966), Chapter 10.

\bibitem{DeMicheli3}
E. De Micheli and G.A. Viano,
Metric and probabilistic information associated with Fredholm integral equations of the first kind,
J. Integral Equations Appl. \textbf{14}, 283--309 (2002).

\bibitem{Prosser}
R.T. Prosser,
The $\varepsilon$--entropy and $\varepsilon$--capacity of certain time--varying channels,
J. Math. Anal. Appl. \textbf{16}, 553--573 (1966).

\bibitem{Gelfand1}
I.M. Gelfand and N.Ya. Vilenkin,
\emph{Generalized Functions IV, Applications of Harmonic Analysis}
(Academic, 1964).

\bibitem{Gelfand2}
I.M. Gelfand and A.M. Yaglom,
Calculation of the amount of information about a random function contained in another such function,
Amer. Math. Soc. Trans. \textbf{12}, 199--246 (1959).

\bibitem{Kato}
T. Kato,
\emph{Perturbation Theory for Linear Operators}
(Springer, 1966), p. 522.

\bibitem{Gori}
F. Gori and G. Guattari,
Shannon number and degrees of freedom of an image,
Opt. Commun. \textbf{7}, 163--165 (1973).

\bibitem{Bertero2}
M. Bertero and C. De Mol,
Super--resolution by data inversion,
in \emph{Progress in Optics}, Vol. XXXVI, E. Wolf, ed. (North--Holland, 1996), pp. 129--178.

\bibitem{Jerri}
A.J. Jerri,
The Shannon sampling theorem -- its various extensions and applications: A tutorial review,
Proc. IEEE \textbf{65}, 1565--1596 (1977).

\end{thebibliography}
\end{document}